\documentclass[11pt]{article}
\usepackage{fullpage}
\usepackage{setspace}
\usepackage{amssymb, amsmath, amsthm}
\usepackage{natbib}
\newcommand{\argmin}{\operatorname*{arg \ min}}

\theoremstyle{plain}

\newtheorem{prop}{Proposition}

\title{Indirect multivariate response linear regression}

\author{
  Aaron J. Molstad \\
  {\small School of Statistics}\\
  {\small University of Minnesota}\\
  {\small \texttt{molst029@umn.edu}}
  \and
  Adam J. Rothman\\
  {\small School of Statistics}\\
  {\small University of Minnesota}\\
  {\small \texttt{arothman@umn.edu}}\\
}

\date{July 16, 2015}

\begin{document}

\setlength{\pdfpageheight}{\paperheight}
\setlength{\pdfpagewidth}{\paperwidth}
 
\maketitle

\begin{abstract}
We propose a new class of estimators of the multivariate response linear regression 
coefficient matrix that exploits the assumption that the response and 
predictors have a joint multivariate Normal distribution. This allows 
us to indirectly estimate the regression coefficient matrix through 
shrinkage estimation of the parameters of the inverse regression, 
or the conditional distribution of the predictors given the responses. 
We establish a convergence rate bound for estimators in our class
and we study two examples. The first example estimator exploits an 
assumption that the inverse regression's coefficient matrix is sparse. 
The second example estimator exploits an assumption that the inverse regression's coefficient 
matrix is rank deficient. These estimators do not require the popular 
assumption that the forward regression coefficient matrix is sparse or has 
small Frobenius norm.  Using simulation studies, we show that our example estimators
outperform relevant competitors for some data generating models. 
\end{abstract}

\section{Introduction} \label{sec:intro}
Some statistical applications require the modeling of a multivariate response.
Let $y_i \in \mathbb{R}^q$ be the measurement of the $q$-variate response for 
the $i$th subject and let $x_i \in \mathbb{R}^p$ be the nonrandom values 
of the $p$ predictors for the $i$th subject ($i=1,\dots,n)$.
The multivariate response linear regression model assumes that $y_i$ is a realization
of the random vector
\begin{equation} \label{normalmodel}
Y_i = \mu_{*} + \beta_{*}' x_i + \varepsilon_i, \quad i=1,\ldots, n,
 \end{equation}
where $\mu_{*} \in \mathbb{R}^q$ is the unknown intercept, 
$\beta_{*}$ is the unknown $p$ by $q$ regression coefficient matrix,
and $\varepsilon_1, \ldots, \varepsilon_{n}$ are independent copies of a mean zero
random vector with covariance matrix $\Sigma_{*E}$.  

The ordinary least squares estimator of $\beta_{*}$ is 
\begin{equation} \label{olsopt}
\hat \beta^{({\rm OLS})} = \argmin_{\beta\in\mathbb{R}^{p\times q}} \|\mathbb{Y} - \mathbb{X}\beta\|_F^2, 
\end{equation}
where $\|\cdot\|_F$ is the Frobenius norm, $\mathbb{R}^{p\times q}$ is the set of real valued $p$ by $q$
matrices, $\mathbb{Y}$ is the $n$ by $q$ matrix with $i$th row 
$(Y_i - n^{-1}\sum_{i=1}^n Y_i)'$, and $\mathbb{X}$ is the $n$ by $p$ matrix with 
$i$th row $(x_i - n^{-1}\sum_{i=1}^n x_i)'$ ($i=1,...,n$).  
It is well known that $\hat \beta^{({\rm OLS})}$ is the maximum likelihood estimator of $\beta_{*}$
when $\varepsilon_1, \ldots, \varepsilon_{n}$ are independent and identically distributed  $N_{q}(0,\Sigma_{*E})$ and 
the corresponding maximum likelihood estimator of $\Sigma_{*E}^{-1}$ exists.

Many shrinkage estimators of $\beta_{*}$ have been proposed by 
penalizing the optimization in \eqref{olsopt}.  Some of these estimators simultaneously
estimate $\beta_{*}$ and remove irrelevant predictors 
\citep{turlach2005simultaneous, obozinski2010joint, peng2010regularized}.
Others encourage an estimator of reduced rank 
\citep{yuan2007dimension, chen2012sparse}.

Under the restriction that $\varepsilon_1, \ldots, \varepsilon_{n}$ are independent 
and identically distributed $N_{q}(0,\Sigma_{*E})$,
shrinkage estimators of $\beta_{*}$ that penalize or constrain the minimization of the
negative loglikelihood have been proposed.  These methods simultaneously estimate 
$\beta_{*}$ and $\Sigma_{*E}^{-1}$.  Examples include 
maximum likelihood reduced rank regression 
\citep{izenman1975reduced, reinsel1998multivariate}, 
envelope models \citep{clc:09, sucook11,sucook12, sucook13},
and multivariate regression 
with covariance estimation \citep{rothman2010sparse,  lee12, bhadra2013joint}.

To fit \eqref{normalmodel} with these shrinkage estimators, 
one exploits explicit assumptions about $\beta_{*}$, but these may be unreasonable
in some applications.  As an alternative, we propose an indirect method
to fit \eqref{normalmodel} without making explicit assumptions about $\beta_{*}$.
We exploit the assumption that response and predictors have a joint 
multivariate Normal distribution and we employ shrinkage estimators of the 
parameters of the conditional distribution of the
predictors given the response. Our method provides an alternative indirect
estimator of $\beta_{*}$, which may be suitable when the existing shrinkage
estimators are inadequate.

\section{A new class of indirect estimators of $\beta_{*}$}
\subsection{Class definition}
We assume that the measured predictor and response pairs
$(x_1, y_1),\dots,(x_n, y_n)$  are a realization of $n$ 
independent copies of $(X,Y)$, where
$(X',Y')' \sim N_{p+q}(\mu_{*}, \Sigma_{*})$. 
We also assume that $\Sigma_{*}$ positive definite.
Define the marginal parameters through the following partitions:
$$
\mu_{*} = 
\left(\begin{array}{c}
\mu_{*X}\\
\mu_{*Y}
      \end{array}
\right),
\quad
\Sigma_{*}=
\left(\begin{array}{cc}
       \Sigma_{*XX}& \Sigma_{*XY}\\
       \Sigma_{*XY}'& \Sigma_{*YY}
      \end{array}
\right).
$$
Our goal is to estimate
the multivariate regression coefficient matrix
$\beta_{*} = \Sigma_{*XX}^{-1} \Sigma_{*XY}$ in the forward regression model
$$
(Y|X=x) \sim N_q( \mu_{*Y} + \beta_{*}'(x-\mu_{*X}), \Sigma_{*E} ),
$$
without assuming that $\beta_{*}$ is sparse or that $\|\beta_{*}\|_{F}^2$ is small.
To do this we will estimate the inverse regression's coefficient matrix 
$\eta_{*} = \Sigma_{YY}^{-1} \Sigma_{XY}'$ and the 
inverse regression's error precision matrix $\Delta_{*}^{-1}$ in
the inverse regression model
$$
(X|Y=y) \sim N_{p}( \mu_{*X} + \eta_{*}'(y-\mu_{*Y}) , \Delta_{*}).
$$
We connect the parameters of
 the inverse regression model to $\beta_{*}$ with the following proposition.
\begin{prop}\label{main.prop}
If $\Sigma_{*}$ is positive definite, then
\begin{equation} 
\beta_* = \Delta_*^{-1}\eta_*^\prime\left(\Sigma_{*YY}^{-1} + \eta_*\Delta_*^{-1}\eta_*^\prime\right)^{-1}.
\end{equation}
\end{prop}
We prove Proposition \ref{main.prop} in Appendix \ref{proofs}. 
This result leads us to propose a class of 
estimators of $\beta_*$ defined by  
\begin{equation}\label{inv_estimator}
\hat{\beta} = \hat{\Delta}^{-1}\hat{\eta}^\prime(\hat{\Sigma}_{YY}^{-1} 
+ \hat{\eta}\hat{\Delta}^{-1}\hat{\eta}^\prime)^{-1}, 
\end{equation}
where $\hat{\eta}$, $\hat{\Delta}$, and $\hat{\Sigma}_{YY}$ are
user-selected estimators of $\eta_*$, $\Delta_*$, and $\Sigma_{*YY}$. 
If $n > \max(p, q)$ and the ordinary sample estimators 
are used for $\hat{\eta}$, $\hat{\Delta}$ and $\hat{\Sigma}_{YY}$, then
$\hat{\beta}$ is equivalent to $\hat \beta^{({\rm OLS})}$. 

We propose to use shrinkage estimators of $\eta_*$, $\Delta_{*}^{-1}$, and $\Sigma_{*YY}^{-1}$
in \eqref{inv_estimator}.  This gives us the potential to indirectly fit an unparsimonious forward
regression model by fitting a parsimonious inverse regression model.
For example, suppose that $\eta_*$ and $\Delta_{*}^{-1}$ are sparse, 
but $\beta_*$ is dense. To fit the inverse regression model, we could 
use any of the forward regression shrinkage estimators discussed in Section \ref{sec:intro}.

\subsection{Related work}
\citet{lee12} proposed an estimator of $\beta_{*}$ that also exploits 
the assumption that  $(X', Y')'$ is multivariate Normal; however,  unlike our approach
that makes no explicit assumptions about $\beta_*$,
their approach assumes that both $\Sigma_{*}^{-1}$ and $\beta_{*}$ are sparse.

Modeling the inverse regression is a well-known idea in multivariate analysis.  For example,
when $Y$ is categorical, quadratic discriminant analysis models $(X|Y=y)$ 
as $p$-variate Normal.  There are also many examples of modeling the inverse regression  
in the sufficient dimension reduction literature \citep{adragnicook2009}.

The most closely related work to ours is that by \citet{cook2013prediction}. 
They proposed indirect estimators of $\beta_{*}$ based on modeling 
the inverse regression in the special
case when the response is univariate, i.e. $q=1$.  Under the same multivariate
Normal assumption on $(X', Y')'$ that we make, 
\citet{cook2013prediction} showed that 
\begin{equation}\label{cfr}
\beta_{*} = \frac{1}{1+\Sigma_{*XY}'\Delta_{*}^{-1}\Sigma_{*XY}/\Sigma_{*YY}}\Delta_{*}^{-1} \Sigma_{*XY}.
\end{equation}
They proposed estimators of $\beta_{*}$ by replacing 
$\Sigma_{*XY}$ and $\Sigma_{*YY}$ in the right hand side of \eqref{cfr} with
their usual sample estimators, and by replacing $\Delta_{*}^{-1}$ with a shrinkage
estimator.  This class of estimators was designed to exploit an abundant
signal rate in the forward univariate response regression when $p > n$. 

\section{Asymptotic Analysis}
We present a convergence rate bound 
for the indirect estimator of $\beta_{*}$
defined by \eqref{inv_estimator}.  
Our bound allows $p$ and $q$ to grow with the sample size $n$. 
In the following proposition, $\| \cdot \|$ is the spectral norm
and $\varphi_{\min}(\cdot)$ is the minimum eigenvalue.

\begin{prop}\label{Proposition1} 
Suppose that following conditions are true:
(i)  $\Sigma_{*}$ is positive definite for all $p+q$;
(ii) the estimator $\hat{\Sigma}_{YY}^{-1}$ is positive definite
for all $q$; 
(iii) the estimator $\hat\Delta^{-1}$ is positive definite for all $p$;
(iv) there exists a positive constant $K$ such that $\varphi_{\min}(\Sigma_{*YY}^{-1}) \geq K$
for all $q$;
and 
(v) there exist sequences $\{a_n\}, \{b_n\}$ and $\{c_n\}$ such that 
$\|\hat\eta - \eta_*\| = O_P(a_n)$,
$\|\hat{\Delta}^{-1} - \Delta_*^{-1} \| = O_P(b_n)$,
$\|\hat\Sigma_{YY}^{-1}  - \Sigma_{*YY}^{-1}\| = O_P(c_n)$,
and $a_n \|\eta_*\| \cdot \|\Delta_*^{-1}\| 
+ b_n  \|\eta_*\|^2  
+ c_n \rightarrow 0$ as $n\rightarrow \infty$.
Then 
\begin{align*}
\|\hat\beta - \beta_*\| = 
O_P \left(
a_n \|\eta_*\|^2 \|\Delta_*^{-1}\|^2 
+ b_n  \|\eta_*\|^3\|\Delta_*^{-1}\|  
+ c_n \|\eta_*\| \cdot \|\Delta_*^{-1}\|   \right).
\end{align*}
\end{prop}
We prove Proposition \ref{Proposition1} in Appendix \ref{proofs}. We used the
spectral norm because it is compatible with the convergence rate bounds established 
for sparse inverse covariance estimators \citep{rothman2008sparse, lam07, ravi08}. 

If the inverse regression is parsimonious in the sense that $\|\eta_{*}\|$
and $\|\Delta_{*}^{-1}\|$ are bounded, then the bound in Proposition \ref{Proposition1} simplifies to 
$\|\hat\beta - \beta_*\| = O_P( a_n + b_n + c_n)$.  From an asymptotic perspective, it is not surprising 
that the indirect estimator of $\beta_{*}$ is only as good as its worst plug-in estimator.  We explore 
finite sample performance in Section \ref{section:sims}.

\section{Example estimators in our class}

\subsection{Sparse inverse regression}\label{sparsesec}
We now describe an estimator of the forward regression coefficient matrix 
$\beta_*$ defined by \eqref{inv_estimator}
that exploits zeros in the inverse regression's coefficient matrix $\eta_{*}$,
zeros in the inverse regression's error precision matrix $\Delta_{*}^{-1}$, 
and zeros in the precision matrix of the responses $\Sigma_{*YY}^{-1}$.  
We estimate $\eta_{*}$ with
\begin{equation} \label{eq:unilasso}
\hat\eta^{{\rm L1}} = \argmin_{\eta\in\mathbb{R}^{q\times p}} 
\left\{ \|\mathbb{X} - \mathbb{Y}\eta\|_{F}^2 
+ \sum_{j=1}^p \lambda_j \sum_{m=1}^q |\eta_{mj}| \right\},
\end{equation}
which separates into $p$ $L_1$-penalized least-squares regressions \citep{tibs96}: the 
first predictor regressed on the response through the $p$th predictor regressed on 
the response.  We select $\lambda_j$ with 5-fold cross-validation, minimizing 
squared prediction error totaled over the folds, 
in the regression of the $j$th predictor on the response 
$(j=1,\ldots, p)$. This allows us to estimate the columns of $\eta_{*}$ in parallel.

We estimate $\Delta_{*}^{-1}$  and $\Sigma_{*YY}^{-1}$ with $L_1$-penalized
Normal likelihood precision matrix estimation \citep{yuan2007model, banerjee06}.
Let $\hat\Sigma_{\gamma, S}^{-1}$ be a generic version of this estimator
with tuning parameter $\gamma$ and input $p$ by $p$ sample covariance matrix $S$:
\begin{equation} \label{glasso}
\hat\Sigma_{\gamma, S}^{-1} = \argmin_{\Omega\in\mathbb{S}_{+}^{p}} 
\left\{ {\rm tr}(\Omega S) - \log |\Omega| + 
\gamma\sum_{j\neq k} |\omega_{jk}|\right\}, 
\end{equation}
where $\mathbb{S}_{+}^{p}$ is the set of symmetric and positive definite $p$ by $p$ matrices.
There are many algorithms that solve \eqref{glasso}. 
Two good choices are the graphical lasso algorithm \citep{yuan2008, fht08} 
and the QUIC algorithm \citep{hsieh11}.  
We select $\gamma$ with
5-fold cross-validation maximizing a validation 
likelihood criterion \citep{huang06}:
\begin{equation} \label{deltatune}
\hat\gamma = \argmin_{\gamma\in\mathcal{G}} 
\sum_{k=1}^5 \left\{ {\rm tr}\left(\hat\Sigma_{\gamma, S_{(-k)}}^{-1} S_{(k)} \right)
 - \log \left| \hat\Sigma_{\gamma, S_{(-k)}}^{-1} \right| \right\},
\end{equation}
where $\mathcal{G}$ is 
a user-selected finite subset of the non-negative real line,
$S_{(-k)}$ is the sample covariance matrix from the observations outside
the $k$th fold, and $S_{(k)}$ is the sample covariance matrix from the observations
in the $k$th fold centered by the sample mean of the observations outside the $k$th fold.
We estimate $\Delta_{*}^{-1}$ using \eqref{glasso}
with its tuning parameter selected by \eqref{deltatune}
and $S = (\mathbb{X} - \mathbb{Y}\hat\eta^{{\rm L1}})'
(\mathbb{X} - \mathbb{Y}\hat\eta^{{\rm L1}})/n$.
Similarly, we estimate $\Sigma_{*YY}^{-1}$  using \eqref{glasso}
with its tuning parameter selected by \eqref{deltatune}
and $S = \mathbb{Y}'\mathbb{Y}/n$.

\subsection{Reduced rank inverse regression}\label{redranksec}
We propose indirect estimators of $\beta_{*}$ that exploit
the assumption that the inverse regression's coefficient matrix $\eta_{*}$ is rank deficient.
We have the following simple proposition that links rank deficiency in $\eta_{*}$ and its estimator
to $\beta_{*}$ and its indirect estimator.
\begin{prop} \label{prop.rr.pop}
If $\Sigma_{*}$ is positive definite, then ${\rm rank}(\beta_*)= {\rm rank}(\eta_*)$.  
In addition, if $\hat{\Sigma}_{YY}^{-1}$ and $\hat{\Delta}^{-1}$ are positive
definite in the indirect estimator $\hat\beta$ defined by \eqref{inv_estimator},
then ${\rm rank}(\hat{\beta}) = {\rm rank}(\hat{\eta})$.
\end{prop}
\noindent
The proof of this proposition is simple so we excluded it to save space. 

We propose the following two example reduced rank indirect estimators of $\beta_{*}$:
\begin{enumerate}
 \item Estimate $\Sigma_{*YY}$ with $\mathbb{Y}'\mathbb{Y}/n$
 and estimate $(\eta_{*}, \Delta_{*}^{-1})$ with Normal 
likelihood reduced rank inverse regression:
\begin{align} \label{inv.RR}
(\hat{\eta}^{(r)}, \hat\Delta^{-1(r)}) 
= \argmin_{(\eta, \Omega) \in\mathbb{R}^{q\times p}\times \mathbb{S}_{+}^p} & 
\left[
n^{-1}{\rm tr}\left\{(\mathbb{X} - \mathbb{Y}\eta)'(\mathbb{X} - \mathbb{Y}\eta)\Omega\right\}
 - \log {\rm det} (\Omega )
\right] \\
& {\rm subject} \ {\rm to} \ {\rm rank}(\eta) = r, \nonumber
\end{align}
where $r$ is selected from $\{0,\ldots, \min(p, q)\}$.  The solution to the optimization
in \eqref{inv.RR} is available in closed form \citep{reinsel1998multivariate}.  
\item Estimate $\eta_{*}$ with $\hat{\eta}^{(r)}$ defined in \eqref{inv.RR},   
estimate  $\Sigma_{*YY}^{-1}$ with \eqref{glasso}
using $S = \mathbb{Y}'\mathbb{Y}/n$, and estimate 
$\Delta_*^{-1}$ with \eqref{glasso}
using
$S = (\mathbb{X} - \mathbb{Y}\hat\eta^{(r)} )^\prime (\mathbb{X} - \mathbb{Y}\hat\eta^{(r)} )/n$.
\end{enumerate}
Both example indirect reduced rank estimators of $\beta_{*}$ are formed by plugging
in the estimators of $\eta_{*}, \Delta_{*}^{-1}$, and $\Sigma_{*YY}$ 
to \eqref{inv_estimator}.  The first 
estimator is likelihood-based and the second estimator exploits 
sparsity in $\Sigma_{*YY}^{-1}$ and $\Delta_{*}^{-1}$.
Neither estimator is defined when $\min(p, q) > n$.  In this case, which we do not address,
a regularized reduced rank estimator of $\eta_*$ could be used instead of 
the estimator defined in \eqref{inv.RR},
e.g. the factor estimation and selection estimator \citep{yuan2007dimension} or 
the reduced rank ridge regression estimator \citep{mukherjee2011reduced}.

\section{Simulations} \label{section:sims}
\subsection{Sparse inverse regression simulation}\label{sparsesim}
We compared the following indirect estimators of $\beta_{*}$ when the inverse regression's
coefficient matrix $\eta_{*}$ is sparse:
\begin{list}{}{}
\item[$I_{L1}$.] This is the indirect estimator proposed in Section \ref{sparsesec}.  
\item[$I_{S}$.] This is an indirect estimator defined by \eqref{inv_estimator}  
with $\hat{\eta}$ defined by \eqref{eq:unilasso}, 
$\hat{\Sigma}_{YY} = \mathbb{Y}^\prime\mathbb{Y}/n$, 
and $\hat{\Delta} = (\mathbb{X} - \mathbb{Y}\hat{\eta}^{L1})^\prime(\mathbb{X} - \mathbb{Y}\hat{\eta}^{L1})/n$. 
\item[$O_{\Delta}$.] This is a part oracle indirect estimator defined by \eqref{inv_estimator} 
with $\hat{\eta}$ defined by \eqref{eq:unilasso}, 
$\hat{\Sigma}_{YY}^{-1}$ defined by  \eqref{glasso}, 
and $\hat{\Delta}^{-1} = \Delta^{-1}_*$.
\item[$O$.] This is a part oracle indirect estimator defined by \eqref{inv_estimator}  
with $\hat{\eta}$ defined by \eqref{eq:unilasso}, 
$\hat{\Sigma}_{YY}^{-1} = \Sigma_{*YY}^{-1}$, 
and $\hat{\Delta}^{-1} = \Delta^{-1}_*$.
\item[$O_{Y}$.]  This is a part oracle indirect estimator defined by \eqref{inv_estimator} 
with $\hat{\eta}$ defined by \eqref{eq:unilasso}, 
$\hat{\Sigma}_{YY}^{-1} = \Sigma_{*YY}^{-1}$, 
and $\hat{\Delta}^{-1}$ defined by  \eqref{glasso}.
\end{list}
We also included the following forward regression estimators of $\beta_{*}$:
\begin{list}{}{}
\item[OLS/MP.] This is the ordinary least squares estimator defined by
$\argmin_{\beta\in\mathbb{R}^{p\times q}} \|\mathbb{Y} - \mathbb{X}\beta\|_{F}^2$.
When $n \leq p$, we use the solution $\mathbb{X}^{-}\mathbb{Y}$, where $\mathbb{X}^{-}$ 
is the Moore-Penrose generalized inverse of $\mathbb{X}$.
\item[R.] This is the ridge penalized least squares estimator defined by
$$
\argmin_{\beta\in\mathbb{R}^{p\times q}}\left( \|\mathbb{Y} - \mathbb{X}\beta\|_{F}^2 + \lambda \|\beta\|_{F}^2 \right).
$$
\item[$\ell_2$.] This is an alternative ridge penalized least squares estimator defined by 
$$
\argmin_{\beta\in\mathbb{R}^{p\times q}}\left( \|\mathbb{Y} - \mathbb{X}\beta\|_{F}^2 
+\sum_{m=1}^q \lambda_m \sum_{j=1}^p \beta_{jm}^2\right),
$$
where a separate tuning parameter is used for each response.
\end{list}
We selected the tuning parameters for uses of \eqref{eq:unilasso} with 5-fold cross-validation, minimizing validation prediction
error on the inverse regression.  Tuning parameters for $\ell_2$ and R were selected with 5-fold cross-validation, 
minimizing validation prediction error on the forward regression.   
We selected tuning parameters for uses of \eqref{glasso} with \eqref{deltatune}.
The candidate set of tuning parameters was  $\left\{10^{-8}, 10^{-7.5}, \dots, 10^{7.5},10^8\right\}$.

For 50 independent replications, we generated a realization of $n$ independent copies of 
$(X^\prime, Y^\prime)^\prime$, where $Y \sim N_q(0, \Sigma_{*YY})$ 
and $(X|Y=y) \sim N_p(\eta_*^\prime y, \Delta_*)$. 
The $(i,j)$th entry of $\Sigma_{*YY}$ was set to $\rho_Y^{|i-j|}$ 
and the $(i,j)$th entry of $\Delta_*$ was set to $\rho_{\Delta}^{|i-j|}$.
We set $\eta_* = Z \circ A$, where $\circ$ denotes the element-wise product:
$Z$ had entries independently drawn from $N(0,1)$ and $A$ had entries independently drawn 
from the Bernoulli distribution with nonzero probability $s_*$. 
This model is ideal for $I_{L1}$ because $\Delta_{*}^{-1}$ and $\Sigma_{*YY}^{-1}$
are both sparse.  Every entry in the corresponding randomly generated $\beta_*$ 
is nonzero with high probability, but the magnitudes of these
entries are small.  This motivated us to compare our indirect estimators of $\beta_{*}$
to the ridge-penalized least squares forward regression estimators R and $\ell_2$.

We evaluated performance with  
model error \citep{breiman1997predicting, yuan2007dimension}, 
which is defined by $\| \Sigma_{*XX}^{1/2} (\hat{\beta} - \beta_*)\|_F^2$.

\begin{table}[ht!]
\caption{ \label{table1} Averages of model error from 50 replications when $n=100,$ $p=20$, and $q=20$. 
All standard errors were less than or equal to 0.05.}
\centering
\begin{tabular}{ r r r |cccccccc}
  \hline
 $\rho_{Y}$ &  $\rho_{\Delta}$ &$s_*$ & $I_{L1}$ & $O$ & $O_{\Delta}$ & $O_{Y}$ & $I_S$ & OLS & $\ell_2$ & R  \\ 
  \hline
  0.7 & 0.0 & 0.1 & 0.61 & 0.32 & 0.53 & 0.40 & 1.35 & 2.10 & 1.23 & 1.22  \\
 0.7 & 0.5 & 0.1  & 0.72 & 0.39 & 0.59 & 0.51 & 1.30 & 1.91 & 1.29 & 1.30 \\ 
 0.7 & 0.7 & 0.1 & 0.76 & 0.45 & 0.65 & 0.56 & 1.27 & 1.73 & 1.27 & 1.29  \\
  0.7 & 0.9 & 0.1 & 0.83 & 0.66 & 0.85 & 0.64 & 1.26 & 1.35 & 1.05 & 1.09 \\
   \hline
  0.0 & 0.9 & 0.1 & 0.81 & 0.87 & 0.87 & 0.79 & 2.04 & 2.34 & 1.26 & 1.87 \\ 
 0.5 & 0.9 & 0.1& 0.96 & 0.76 & 0.99 & 0.74 & 1.63 & 1.84 & 1.36 & 1.49  \\ 
 0.9 & 0.9 & 0.1 & 0.46 & 0.39 & 0.47 & 0.36 & 0.63 & 0.62 & 0.48 & 0.48 \\ 
   \hline
   0.7 & 0.9 & 0.3& 0.60 & 0.53 & 0.65 & 0.46 & 0.83 & 0.67 & 0.64 & 0.63 \\ 
 0.7 & 0.9 & 0.5 & 0.48 & 0.37 & 0.48 & 0.37 & 0.65 & 0.53 & 0.52 & 0.51  \\ 
  0.7 & 0.9 & 0.7 & 0.42 & 0.29 & 0.39 & 0.31 & 0.55 & 0.46 & 0.45 & 0.44  \\ 
 \hline
\end{tabular}
\end{table}
We report the average model errors, based on these 50 replications, in Table \ref{table1}.
When $s_{*}=0.1$, the indirect estimators defined by \eqref{inv_estimator} 
performed well for all choices of $\rho_Y$ and $\rho_{\Delta}$. 
Our proposed estimator $I_{L1}$ was competitive 
with other indirect estimators also defined by \eqref{inv_estimator}, even those that
used some oracle information.  As $s_*$ increased with $\rho_Y = 0.7$ and $\rho_{\Delta} = 0.9$ 
fixed, the forward regression estimators performed nearly as well as $I_{L1}$.

\begin{table}[ht!]
\centering
\caption{\label{table2} 
Averages of model error from 50 replications when $n=50,$ $p=60$, and $q=60$. 
All standard errors were 0.69 or less,
except for MP, which had standard errors between 0.77 and 3.16.}
\begin{tabular}{ r r r |ccccccc}
  \hline
 $\rho_{Y}$ &  $\rho_{\Delta}$ &$s_*$ &  $I_{L1}$ & $O$  & $O_{\Delta}$ & $O_Y$ & MP & $\ell_2$ & R \\ 
    \hline
  0.7 & 0.0 & 0.1 & 8.59 & 4.28 & 5.70 & 7.40 & 78.33 & 13.85 & 12.44 \\ 
 0.7 & 0.5 & 0.1  & 9.67 & 5.09 & 6.37 & 8.49 & 73.82 & 14.79 & 13.34 \\ 
  0.7 & 0.7 & 0.1 & 10.01 & 6.37 & 7.44 & 8.75 & 70.30 & 15.56 & 14.40 \\ 
  0.7 & 0.9 & 0.1  & 9.92 & 10.07 & 11.44 & 8.88 & 61.83 & 16.43 & 15.94 \\ 
   \hline
  0.0 & 0.9 & 0.1 & 15.17 & 17.09 & 16.93 & 15.23 & 119.60 & 28.63 & 29.41 \\  
    0.5 & 0.9 & 0.1 & 14.88 & 13.59 & 16.91 & 12.01 & 86.88 & 23.62 & 22.69 \\ 
 0.9 & 0.9 & 0.1 & 4.71 & 4.78 & 5.94 & 3.99 & 25.37 & 6.36 & 5.91 \\ 
   \hline
 0.7 & 0.9 & 0.3 & 16.86 & 17.43 & 19.66 & 15.44 & 43.88 & 15.30 & 14.14 \\ 
  0.7 & 0.9 & 0.5 & 26.89 & 26.81 & 29.93 & 24.95 & 36.87 & 14.79 & 13.62 \\ 
  0.7 & 0.9 & 0.7 & 31.86 & 35.98 & 38.64 & 30.36 & 33.58 & 14.35 & 13.65 \\ 
\hline
\end{tabular}
\end{table}
Similarly, Table \ref{table2} shows that when $s_*=0.1$, $I_{L1}$ outperforms 
all three forward regression estimators. However, unlike in the lower dimensional setting illustrated in table \ref{table1}, 
when $\eta_*$ is not sparse, i.e. $s_* \geq .3$, $I_{L1}$ is outperformed by forward regression approaches.
The part oracle method $O_{Y}$ that used the knowledge of $\Sigma_{*YY}^{-1}$ outperformed 
the other two part oracle indirect estimators $O$ and $O_{\Delta}$
when $\rho_{\Delta}=.9$. Also, when $\rho_{\Delta}=.9$, $I_{L1}$ was competitive with the part oracle estimators. 
Taken together, the results in Tables \ref{table1} and \ref{table2} 
suggest that when $\eta_*$ is very sparse, our proposed indirect estimator $I_{L1}$ 
may perform nearly as well as the part oracle indirect 
estimators and the forward regression estimators.

\subsection{Reduced rank inverse regression simulation}  \label{redranksim1}
We compared the performance of the following indirect reduced rank estimators of $\beta_{*}$: 
\begin{list}{}{}
\item[$I_{ML}^{(r)}$.] This is the likelihood-based indirect example estimator 1 proposed in Section \ref{redranksec}.
\item[$I^{(r)}$.] This is the indirect example estimator 2 proposed in Section \ref{redranksec}, which uses
sparse estimators of $\Sigma_{*YY}^{-1}$ and $\Delta_{*}^{-1}$ in \eqref{inv_estimator}.
\item[$O^{(r)}$.] This is a part oracle indirect estimator defined by \eqref{inv_estimator} 
with $\hat{\eta}$ defined by \eqref{inv.RR}, 
$\hat{\Delta}^{-1} = \Delta_*^{-1}$, and $\hat{\Sigma}_{YY}^{-1} = \Sigma_{*YY}^{-1}$.
\item[$O^{(r)}_{\Delta}$.] 
This is a part oracle indirect estimator defined by \eqref{inv_estimator} 
with $\hat{\eta}$ defined by \eqref{inv.RR}, 
$\hat{\Delta}^{-1}$ defined by \eqref{glasso}, 
and $\hat{\Sigma}_{YY}^{-1} = \Sigma_{*YY}^{-1}$.  
\item[$O^{(r)}_{Y}$.] 
This is a part oracle indirect estimator defined by \eqref{inv_estimator} 
with $\hat{\eta}$ defined by \eqref{inv.RR}, 
$\hat{\Delta}^{-1} = \Delta_*^{-1}$, 
$\hat{\Delta}^{-1}$ defined by \eqref{glasso}, 
and $\hat{\Sigma}_{YY}^{-1} $ defined by \eqref{glasso}.  
\end{list} 
We compared these indirect estimators to the following forward reduced rank regression estimator: 
\begin{list}{}{}
\item[RR.] This is the likelihood based reduced rank regression \citep{izenman1975reduced, reinsel1998multivariate}.  
The estimator of $\beta_{*}$ and the estimator of the forward 
regression's error precision matrix $\Sigma_{*E}^{-1}$ are defined by
\begin{align*} 
(\hat{\beta}^{(r)}, \hat\Sigma_{E}^{-1(r)}) 
= \argmin_{(\beta, \Omega) \in\mathbb{R}^{p\times q}\times \mathbb{S}_{+}^q} & 
\left[
n^{-1}{\rm tr}\left\{(\mathbb{Y} - \mathbb{X}\beta)'(\mathbb{Y} - \mathbb{X}\beta)\Omega\right\}
 - \log {\rm det} (\Omega )
\right] \\
& {\rm subject} \ {\rm to} \ {\rm rank}(\beta) = r.
\end{align*}
\end{list}
We selected the rank parameter $r$ for uses of \eqref{inv.RR} with 5-fold cross-validation, minimizing validation prediction
error on the inverse regression.  The rank parameter for RR was selected with 5-fold cross-validation, 
minimizing validation prediction error on the forward regression.   
We selected tuning parameters for uses of \eqref{glasso} with \eqref{deltatune}.
The candidate set of tuning parameters was  $\left\{10^{-8}, 10^{-7.5}, \dots, 10^{7.5},10^8\right\}$.

For 50 independent replications, we generated a realization of $n$ independent copies 
of $(X', Y')'$ where $Y \sim N_q(0, \Sigma_{*YY})$ and $(X|Y=y) \sim N_p(\eta_*^\prime y, \Delta_*)$. 
The $(i,j)$th entry of $\Sigma_{*YY}$ was set to $\rho_Y^{|i-j|}$ and the $(i,j)$th entry of $\Delta_*$ 
was set to $\rho_{\Delta}^{|i-j|}$. After specifying $r_* \leq \min(p, q)$, we set $\eta_* = PQ$,
where $P \in \mathbb{R}^{q \times r_*}$ and $Q \in \mathbb{R}^{r_* \times p}$ had entries
independently drawn from $N(0,1)$ so that $r_* = \text{rank}(\eta_*) = \text{rank}(\beta_*)$.
As we did in the simulation in Section \ref{sparsesim}, we measured 
performance with model error.

\begin{table}[t!]
\centering
\caption{\label{table.rr.1}
 Averages of model error from 50 replications when $n=100,$ $p=20$, and $q=20$. 
All standard errors were less than or equal to 0.05.}
\begin{tabular}{ r r r |ccccccc}
  \hline
  $\rho_{Y}$ &  $\rho_{\Delta}$ & $r_*$ & $I^{(r)}$ & $O^{(r)}$ & $O^{(r)}_{\Delta}$ & $O^{(r)}_{Y}$ & $I_{ML}^{(r)}$ &  OLS & RR \\ 
 \hline
 0.7 & 0.0 & 10 & 0.33 & 0.04 & 0.86 &  0.75 & 0.64 & 1.38 & 0.64  \\ 
  0.7 & 0.5 & 10 & 0.34 & 0.04 & 0.86 & 0.74 & 0.60 & 1.31 & 0.60 \\ 
  0.7 & 0.7 & 10 & 0.31 & 0.03 & 0.86 & 0.80 & 0.62 & 1.32 & 0.61 \\ 
 0.7 & 0.9 & 10 & 0.31 & 0.02  & 0.85 & 0.88 & 0.60 & 1.30 & 0.61 \\ 
  \hline
 0.0 & 0.9 & 10 & 0.15 &0.03  &1.00 & 1.77 & 1.22 & 2.61 & 1.21 \\ 
  0.5 & 0.9 & 10 & 0.42 & 0.01  &1.11 & 1.36 & 0.90 & 1.97 & 0.89  \\ 
  0.9 & 0.9 & 10 & 0.12 & 0.01 & 0.32 & 0.30 & 0.22 & 0.46 & 0.22  \\ 
   \hline
  0.7 & 0.9 & 4 & 0.35 & 0.02 &1.73 & 2.61 & 0.49 & 3.12 & 0.49  \\ 
  0.7 & 0.9 & 8 & 0.35 &0.01  &  1.15 & 1.33 & 0.68 & 1.73 & 0.65  \\ 
 0.7 & 0.9 & 12 & 0.31 & 0.04 & 0.64 & 0.59 & 0.55 & 0.96 & 0.53  \\ 
   0.7 & 0.9 & 16 & 0.25 & 0.08 &  0.30 & 0.20 & 0.44 & 0.50 & 0.42  \\ 
\hline
\end{tabular}
\end{table}

We report the model errors, averaged over the 50 independent replications, in Table \ref{table.rr.1}. 
Under every setting, $I^{(r)}$ outperformed all non-oracle competitors. When $r_* \leq 12$, 
$I^{(r)}$ outperformed both $O_\Delta^{(r)}$ and $O_Y^{(r)}$, which suggests that shrinkage estimation of
$\Delta_*^{-1}$ and $\Sigma_{*YY}^{-1}$ was helpful. In each setting, $I_{ML}^{(r)}$ performed similarly
 to $RR$ even though they are estimating parameters of different condition distributions.  

\subsection{Reduced rank forward regression simulation}
Our simulation studies in the previous sections used inverse regression data generating models.
In this section, we compare the estimators from Section \ref{redranksim1}
using a forward regression data generating model. 

For 50 independent replications, we generated a realization of $n$ independent copies of 
$(X^\prime, Y^\prime)^\prime$ where $X \sim N_p(0, \Sigma_{*XX})$ and 
$(Y|X=x) \sim N_q(\beta_*^\prime x, \Sigma_{*E})$. 
The $(i,j)$th entry of $\Sigma_{*XX}$ was set to $\rho_X^{|i-j|}$ and 
the $(i,j)$th entry of $\Sigma_{*E}$ was set to $\rho_{E}^{|i-j|}$. 
After specifying $r_* \leq \min(p, q)$, we set $\beta_* = ZQ$ where 
$Z \in \mathbb{R}^{p \times r_*}$ had entries independently drawn from $N(0,1)$ 
and $Q \in \mathbb{R}^{r_* \times q}$ had entries independently drawn from 
$\text{Uniform}(-1/4,1/4)$. In this data generating model, 
neither $\Delta_*^{-1}$ nor $\Sigma_{*YY}^{-1}$ had entries equal to zero.

\begin{table}[ht!]
\centering
\caption{\label{for.table}
Averages of model error from 50 replications when $n=100,$ $p=20$, and $q=20$. 
All standard errors were less than
or equal to 0.21.}
\begin{tabular}{ r r r |ccccccc}
  \hline
  $\rho_{X}$ &  $\rho_{E}$ & $r_*$ & $I^{(r)}$ & $O^{(r)}$&  $O^{(r)}_{\Delta}$ & $O^{(r)}_{Y}$ & $I_{ML}^{(r)}$ &  OLS & RR \\ 
 \hline
0.0 & 0.9 & 10 & 2.79 & 0.54 & 4.27 & 5.05 & 2.48 & 4.99 & 2.82 \\
0.5 & 0.9 & 10 & 2.90 & 0.47 & 5.36 & 5.94 & 2.73 & 5.00 & 2.89  \\ 
0.7 & 0.9 & 10 & 2.97 & 0.51&  4.64 & 5.03 & 2.71 & 4.93 & 2.76 \\ 
0.9 & 0.9 & 10 & 2.84 & 0.73 & 3.78 & 4.16 & 2.67 & 5.19 & 2.73  \\ 
   \hline
0.7 & 0.0 & 10 & 4.66 & 1.92   & 3.59 & 5.88 & 4.53 & 5.11 & 4.34 \\ 
0.7 & 0.5 & 10 & 4.27 &1.65  & 3.88 & 5.51 & 3.99 & 5.06 & 3.97  \\
0.7 & 0.7 & 10 & 3.55 & 1.26&  3.99 & 5.29 & 3.43 & 5.00 & 3.44  \\ 
   \hline
0.7 & 0.9 & 4 & 1.27 & 0.08  & 3.84 & 4.71 & 0.95 & 5.00 & 1.11\\
0.7 & 0.9 & 8 & 2.39 & 0.36 & 4.15 & 5.15 & 2.05 & 4.81 & 2.22 \\ 
0.7 & 0.9 & 12 & 3.58 & 0.79 & 4.44 & 5.21 & 3.20 & 5.15 & 3.27  \\ 
 0.7 & 0.9 & 16 & 4.53 & 1.29 & 4.62 & 4.42 & 4.33 & 5.11 & 4.38 \\ 
   \hline
\end{tabular}
\end{table}

The model errors, averaged over the 50 replications, are reported in Table \ref{for.table}. 
 Both $I^{(r)}$ and $I_{ML}^{(r)}$ were competitive with RR in most settings.
 Although neither $\Delta_*^{-1}$ nor $\Sigma_{*YY}^{-1}$ were sparse,
 we again see that $I^{(r)}$ generally outperforms $O_{Y}^{(r)}$ and $O_{\Delta}^{(r)}$, 
 both of which use some oracle information. These results indicate that shrinkage estimators 
 of $\Delta_*^{-1}$ and $\Sigma_{*YY}^{-1}$ in \eqref{inv_estimator} are helpful when neither is sparse. 

\section{Tobacco chemical composition data example}
As an example application, we use the chemical composition of tobacco leaves data from 
\cite{anderson1952statistical} and \cite{izenman2009modern}. These data have $n=25$ cases, $p=6$
predictors, and $q=3$ responses. 
The names of the predictors, taken from page 183 of \citet{izenman2009modern}, are percent nitrogen,
 percent chlorine,
 percent potassium,
 percent phosphorus, 
 percent calcium,
and percent magnesium.
The names of the response variables, also taken from page 183 of \citet{izenman2009modern}, are
rate of cigarette burn in inches per 1,000 seconds,
 percent sugar in the leaf, and
percent nicotine in the leaf. 
In these data, it may inappropriate to assume that $\Delta_*^{-1}$ is sparse. 
For this reason, we consider another example indirect estimator of $\beta_{*}$
called $I_{L2}$ that estimates $\eta_*$ with \eqref{eq:unilasso}, estimates $\Sigma_{*YY}^{-1}$ with 
\eqref{glasso} using $S =\mathbb{Y}^\prime \mathbb{Y}/n$, and estimates $\Delta_*^{-1}$ with
\begin{equation}\label{l2prec}
\argmin_{\Omega \in \mathbb{S}^P_{+}}\left\{\text{tr}(\Omega S) - \text{log } \text{det}\left( \Omega \right) + \gamma \sum_{j,k}|\omega_{jk}|^2 \right\}, 
\end{equation}
where $S = (\mathbb{Y} - \mathbb{X}\hat{\eta}^{L1})^\prime(\mathbb{Y} - \mathbb{X}\hat{\eta}^{L1})/n$. 
We compute \eqref{l2prec} with the closed form solution derived by \cite{witten2009covariance}.  
As before, we select $\gamma$ from $\{10^8, 10^{-7.5},\ldots, 10^{7.5}, 10^8\}$ using (\ref{deltatune}). 
We also consider the forward regression estimators RR, $\ell_2$, and OLS defined in Section \ref{sparsesim} and Section \ref{redranksim1}. 
We introduce another competitor $\ell_1$, defined as
$$
\argmin_{\beta \in \mathbb{R}^{p \times q}}
\left\{ \|\mathbb{Y} - \mathbb{X}\beta\|_F^2 + \sum_{j=1}^q \lambda_j \sum_{l=1}^p |\beta_{jl}|\right\}, 
$$
which is equivalent to performing $q$ separate lasso regressions \citep{tibs96}. 
We randomly split the data into a 40\% test set and  60\% training set in each of 500 replications 
and we measured the squared prediction error on the test set. 
All tuning parameters were chosen from $\{10^8, 10^{-7.5},\ldots, 10^{7.5}, 10^8\}$
by 5-fold cross validation. 

\begin{table}[ht!]
\caption{\label{tobacco.table} Averages of squared prediction error, with standard errors in parenthesis, 
for each response variable from 500 replications.}
\centering
\begin{tabular}{c|cccccccc}
  \hline
& $I^{(r)}$ & $I_{L1}$ & $I_{L2}$ & OLS & RR & $\ell_2$  & $\ell_1$ \\ 
  \hline
  Rate of burn & 1.19 & 1.33 & 0.45 & 2.96  & 2.17 & 0.57 & 1.55 \\ 
  & (0.08) & (0.10) & (0.03) & (0.15)  & (0.15) & (0.07) & (0.13) \\ 
  Percent sugar & 442.38 & 347.76 & 235.55 & 799.03 &  605.30 & 365.13 & 583.98 \\ 
   & (17.97) & (21.31) & (6.31) & (29.45)  & (25.52) & (20.68) & (24.36) \\ 
  Percent nicotene & 2.55 & 2.54 & 0.79 & 5.65  & 4.59 & 0.81 & 2.82 \\ 
   & (0.29) & (0.30) & (0.05) & (0.41) & (0.31) & (0.21) & (0.29) \\ 
\hline
\end{tabular}
\end{table}
Table \ref{tobacco.table} shows squared prediction errors, averaged over
the 10 predictions and the 500 replications.  
These results indicate that $I_{L2}$ outperforms all the competitors we considered. 
Also, $I_{L1}$ was outperformed by $\ell_2$, 
but was competitive with separate lasso regressions. 
Reduced rank regression was not competitive with the proposed indirect estimators. 

\section*{Acknowledgment}
We thank Liliana Forzani for an important discussion.  This research
is supported in part by a grant from the U.S. National Science Foundation.

\appendix
\section{Appendix}\label{app}
\subsection{Proofs}\label{proofs}

\begin{proof}[Proof of Proposition \ref{main.prop}]
Since $\Sigma_{*}$ is positive definite, we apply the partitioned inverse formula to obtain that
\begin{align*}
\Sigma_*^{-1}  
& = 
\left(
\begin{array}{c c}
 \Sigma_{*XX} & \Sigma_{*XY}  \\
 \Sigma_{*XY}' & \Sigma_{*YY} \\
\end{array} \right)^{-1}
 =
\left(
\begin{array}{c c}
 \Delta^{-1}_* & - \beta_* \Sigma_{*E}^{-1}  \\
 -\eta_*\Delta^{-1}_* & \Sigma_{*E}^{-1} \\
\end{array} \right),
\end{align*}
where $\Delta_{*} = \Sigma_{*XX} - \Sigma_{*XY}\Sigma_{*YY}^{-1}\Sigma_{*XY}'$
and $\Sigma_{*E} = \Sigma_{*YY} - \Sigma_{*XY}'\Sigma_{*XX}^{-1}\Sigma_{*XY}$.
The symmetry of $\Sigma_*^{-1}$ implies that $\beta_* \Sigma_{*E}^{-1}  = (\eta_*\Delta^{-1}_*)'$
so 
\begin{equation} \label{betapart1}
\beta_{*} = \Delta^{-1}_* \eta_{*}' \Sigma_{*E}.  
\end{equation}
Using the Woodbury identity, 
\begin{align}
\Sigma_{*E}^{-1} &= (\Sigma_{*YY} - \Sigma^\prime_{*XY}\Sigma_{*XX}^{-1} \Sigma_{*XY})^{-1} \nonumber \\
& = \Sigma_{*YY}^{-1} +  \Sigma_{*YY}^{-1}\Sigma^\prime_{*XY}\left(\Sigma_{*XX}^{-1} 
- \Sigma_{*XY}\Sigma_{*YY}^{-1} \Sigma^\prime_{*XY}\right)^{-1} \Sigma_{XY}\Sigma_{*YY}^{-1} \nonumber \\
& = \Sigma_{*YY}^{-1} + \eta_*\Delta_*^{-1}\eta_{*}'. \label{exprSigmaEinv}
\end{align}
Using the inverse of the expression above in \eqref{betapart1} establishes the result.
\end{proof}

In our proof of Proposition \ref{Proposition1}, we use the 
matrix inequality
\begin{align}\label{identity1}
\|A^{(1)}A^{(2)}A^{(3)} - B^{(1)}B^{(2)}B^{(3)}\| & \leq \sum_{j=1}^3\|A^{(j)} - B^{(j)}\| 
\prod_{k \neq j} \|B^{(k)}\| \notag \\
& + \sum_{j=1}^3 \|B^{(j)}\| \prod_{k \neq j}\|A^{(k)} - B^{(k)}\| + \prod_{j=1}^3\|A^{(j)} - B^{(j)}\|.
\end{align}
\citet{bl06} used \eqref{identity1} to prove their Theorem 3.

\begin{proof}[Proof of Proposition \ref{Proposition1}] 
From \eqref{exprSigmaEinv} in the proof of Proposition \ref{main.prop}, 
$\Sigma_{*E}^{-1} = \Sigma_{*YY}^{-1} + \eta_*\Delta_*^{-1}\eta_*'$.
Define $\hat{\Sigma}_E^{-1} = \hat{\Sigma}_{YY}^{-1} + \hat{\eta}\hat{\Delta}^{-1}\hat{\eta}'$.
Applying \eqref{identity1}, 
\begin{align}
\|\hat\beta - \beta_*\| 
= &  \|\hat{\Delta}^{-1}\hat{\eta}'\hat\Sigma_{E} - \Delta_*^{-1}\eta_*'
 \Sigma_{*E}\|\nonumber \\
\leq &  
\|\hat{\Delta}^{-1} - \Delta_*^{-1}\| \cdot \|\eta_*\| \cdot \|\Sigma_{*E}\|
+ \|\hat{\eta} - \eta_*\| \cdot \|\Delta_*^{-1}\| \cdot \|\Sigma_{*E}\| 
+ \|\hat\Sigma_{E} - \Sigma_{*E}\| \cdot \|\Delta_*^{-1}\| \cdot \|\eta_*\| \nonumber \\
& + \|\Delta_*^{-1}\| \cdot \|\hat{\eta} - \eta_*\| \cdot \|\hat\Sigma_{E} - \Sigma_{*E}\|
 + \|\eta_*\| \cdot \|\hat{\Delta}^{-1} - \Delta_*^{-1}\| \cdot\|\hat\Sigma_{E} - \Sigma_{*E}\| \nonumber\\
& + \|\Sigma_{*E}\| \cdot \|\hat{\Delta}^{-1} - \Delta_*^{-1}\| \cdot \|\hat{\eta} - \eta_*\|
 + \|\hat{\eta} - \eta_*\| \cdot \|\hat{\Delta}^{-1} - \Delta_*^{-1}\|  \cdot
\|\hat\Sigma_{E} - \Sigma_{*E}\|. \label{expandedform} 
\end{align}
We will show that the third term in \eqref{expandedform} dominates the others.  We continue by deriving its bound.
Employing a matrix identity used by \citet{cai08}, we write
$\hat\Sigma_{E} - \Sigma_{*E} = \Sigma_{*E} (\Sigma_{*E}^{-1} - \hat\Sigma_E^{-1})\hat{\Sigma}_E$, 
so
\begin{equation} \label{sigma.e.inv}
\|\hat\Sigma_{E} - \Sigma_{*E}\| \leq 
\|\hat{\Sigma}_E\| \cdot \|\Sigma_{*E}\| \cdot \|\hat{\Sigma}_E^{-1} - \Sigma_{*E}^{-1}\|.
\end{equation}
Using the triangle inequality and \eqref{identity1}, 
\begin{align}
\|\hat{\Sigma}_E^{-1} - \Sigma_{*E}^{-1} \| 
& \leq \|\hat{\Sigma}_{YY}^{-1} - \Sigma_{*YY}^{-1}\| +  
\|\hat{\eta} \hat{\Delta}^{-1}\hat{\eta}' - \eta_*\Delta_*^{-1}\eta_*'\|\notag \\
& \leq  \|\hat{\Sigma}_{YY}^{-1} - \Sigma_{*YY}^{-1}\| + 
2\|\hat{\eta} - \eta_*\| \cdot \|\Delta_*^{-1}\| \cdot \|\eta_*\| + 
\|\hat{\Delta}^{-1} - \Delta_*^{-1}\| \cdot \|\eta_*\|^2 \notag \\ &+ 2\|\eta_*\| \cdot
\|\hat{\Delta}^{-1} - \Delta_*^{-1}\| \cdot \|\hat{\eta} - \eta_*\| +  
\|\Delta_*^{-1}\| \cdot \|\hat{\eta} - \eta_*\|^2  +  \|\hat{\eta} - \eta_*\|^2 
\|\hat{\Delta}^{-1} - \Delta_*^{-1}\| \notag \\
& = O_P\left(c_n + a_n \|\eta_*\| \cdot
\|\Delta_*^{-1}\| + b_n\|\eta_*\|^2\right). \label{d-d_*}
\end{align} 
Since $\varphi_{\min}(\Sigma_{*YY}^{-1})\geq K$ and $\Delta_{*}^{-1}$ is positive definite,
Weyl's eigenvalue inequality implies that 
$\varphi_{\min}(\Sigma_{*E}^{-1})\geq K$  so 
\begin{equation} \label{sigma.e}
\|\Sigma_{*E}\| = \varphi_{\min}^{-1}(\Sigma_{*E}^{-1}) \leq 1/K.  
\end{equation}
Also, 
\begin{equation}\label{d_inverse}
\|\hat{\Sigma}_E\| = \varphi_{\min}^{-1}(\hat{\Sigma}_E^{-1}) = O_P(1).
\end{equation}
because $\varphi_{\min}(\Sigma_{*E}^{-1})\geq K$, 
$\hat{\Sigma}_E$ is positive definite,
and $a_n \|\eta_*\| \cdot \|\Delta_*^{-1}\| 
+ b_n  \|\eta_*\|^2  
+ c_n=o(1)$ in \eqref{d-d_*}.
Using \eqref{d-d_*}, \eqref{sigma.e}, and \eqref{d_inverse}, in \eqref{sigma.e.inv},
$$
\|\hat\Sigma_{E} - \Sigma_{*E}\| = O_P\left(a_n \|\eta_*\| \cdot
\|\Delta_*^{-1}\| + b_n\|\eta_*\|^2 + c_n\right).
$$
We then see that the third term in \eqref{expandedform} dominates and 
\begin{align*}
\|\hat\beta - \beta_*\| & = O_P\left\{ \left(a_n \|\eta_*\| \cdot
\|\Delta_*^{-1}\| + b_n\|\eta_*\|^2 + c_n\right) \|\eta_*\|\|\Delta_*^{-1}\| \right\}\\
& = O_P \left(
a_n \|\eta_*\|^2 \|\Delta_*^{-1}\|^2 
+ b_n  \|\eta_*\|^3\|\Delta_*^{-1}\|  
+ c_n \|\eta_*\| \cdot \|\Delta_*^{-1}\|   \right).
\end{align*}
\end{proof}

\bibliography{Indirect_multiple_response-2015-07-16}

\end{document}